\RequirePackage{fix-cm}
\documentclass[smallextended]{svjour3}       
\smartqed  

\spnewtheorem{definitionitalic}{Definition}{\bf}{\it}
\spnewtheorem*{nonumtheorem}{Theorem}{\bf}{\it}

\usepackage{bbm}
\usepackage{bm}
\usepackage{amsmath,amssymb}
\usepackage{verbatim}
\usepackage{stmaryrd}
\usepackage{graphicx,color}
\usepackage{mathtools}
\usepackage{slashed}
\usepackage[all,cmtip]{xy}
\usepackage{tikz-cd}
\usepackage{array}
\usepackage{enumerate}

\usepackage{hyperref}
\definecolor{lcolor}{rgb}{0.6,0.3,0.3}
\hypersetup{colorlinks=true,linkcolor=lcolor,urlcolor=lcolor,citecolor=lcolor}

\DeclareMathOperator{\Gr}{Gr}

\DeclareMathOperator{\Der}{Der}

\DeclareMathOperator{\CDerEnd}{CDer}

\DeclareMathOperator{\image}{image}

\DeclareMathOperator{\Diff}{Diff}
\DeclareMathOperator{\Aut}{Aut}

\DeclareMathOperator{\Hom}{Hom}

\DeclareMathOperator{\End}{End}

\DeclareMathOperator{\sech}{sech}

\newcommand{\tsum}{\textstyle\sum}

\newcommand{\p}{\partial}

\renewcommand{\epsilon}{USE eps INSTEAD}

\newcommand{\R}{\mathbbm{R}}
\newcommand{\Z}{\mathbbm{Z}}

\renewcommand{\subset}{\subseteq}
\renewcommand{\supset}{\supseteq}
\newcommand{\secref}[1]{Section \ref{sec:#1}}
\newcommand{\appendixref}[1]{Appendix \ref{appendix:#1}}
\newcommand{\lemmaref}[1]{Lemma \ref{lemma:#1}}
\newcommand{\theoremref}[1]{Theorem \ref{theorem:#1}}
\newcommand{\remarkref}[1]{Remark \ref{remark:#1}}
\newcommand{\defref}[1]{Definition \ref{def:#1}}

\newcommand{\httpref}[1]{\href{http://#1}{#1}}

\newcommand{\catstyle}[1]{\textnormal{\textsf{{\fontseries{sbc}\selectfont #1}}}}

\newcommand{\gx}{\mathfrak{g}}

\newcommand{\I}{\mathcal{I}}

\renewcommand{\L}{\mathcal{L}}
\newcommand{\E}{\mathcal{E}}

\newcommand{\MC}{\catstyle{MC}}

\newcommand{\RR}{{C^\infty}}

\newcommand{\px}{\mathcal{P}}
\newcommand{\ax}{\mathcal{A}}
\newcommand{\pxg}{\mathfrak{p}}
\newcommand{\axg}{\mathfrak{a}}
\newcommand{\sbounce}{\textnormal{bounce}}
\newcommand{\sfree}{\textnormal{free}}
\newcommand{\pbounce}{\px_{\sbounce}}
\newcommand{\pfree}{\px_{\sfree}}
\newcommand{\abounce}{\ax_{\sbounce}}
\newcommand{\afree}{\ax_{\sfree}}

\newcommand{\ww}{{\wedge W}}

\newcommand{\inj}{\hookrightarrow}
\newcommand{\surj}{\twoheadrightarrow}

\newcommand{\sv}[1]{\mathbf{s}_{#1}}

\newcommand{\mcp}{\MC_{\textnormal{pre}}(\afree)}
\newcommand{\mcn}{\MC_{\textnormal{normal}}(\afree)}

\newcommand{\transition}{\chi}

\renewcommand{\next}[1]{\smash{\slashed{#1}}}

\newcommand{\nmlz}{\mathcal{N}}

\newcommand{\citeone}{\cite{rt1}}

\newcommand{\citethree}{\cite{rt3}}

\newcommand{\parsp}{\mathcal{G}}
\newcommand{\datasp}{\mathcal{D}}

\newcommand{\spatialm}{S}

\usepackage{mathptmx}      
\begin{document}

\title{Filtered expansions in general relativity I}


\author{Michael Reiterer \and
        Eugene Trubowitz
}


\institute{Michael Reiterer,
              ETH Zurich,
              \email{michael.reiterer@protonmail.com}             \\
            \rule{55pt}{0pt} \text{Present affiliation: The Hebrew University of Jerusalem}\\
           Eugene Trubowitz,
              ETH Zurich,
              \email{eugene.trubowitz@math.ethz.ch}
}

\date{Received: date / Accepted: date}

\maketitle

\begin{abstract}
This is the first of two papers
in which
we construct formal power series solutions
in external parameters
to the vacuum Einstein equations,
implementing one bounce for the Belinskii-Khalatnikov-Lifshitz (BKL) proposal
for spatially inhomogeneous spacetimes.
We use a graded Lie algebra, homological framework.
A dedicated filtration encodes key features of the BKL proposal,
and we use it to set up an unobstructed perturbative problem.

\end{abstract}

\tableofcontents

\newcommand{\tabdict}{%
  \begin{table}
  \begin{center}
    \begin{tabular}{|c|c|}
      \hline
      this paper & closest analogue in the BKL papers\\
      \hline
      $\MC(\E) = \{x \in \E^1 \mid [x,x] = 0 \}$ & (possibly singular) Ricci-flat metrics \\
      $\E$ & --\\
      $F_0 \E^1$ & spatial derivatives ignored to leading order\\
      $F_\alpha \E^1$ & smallness hierarchy\\
      $F_\alpha \E$ & -- \\
      $\pfree=$ Rees algebra for 3-index $F_\alpha \E$ & -- \\
      $\pbounce=$ Rees algebra for 2-index $F_\alpha \E$ & -- \\
      $\MC(\afree)$, with
      $\afree = \Gr \pfree$ & generalized Kasner spacetimes\\
      $\MC(\abounce)$,
      with $\abounce = \Gr \pbounce$ & leading term for one bounce spacetimes\\
      $\MC(\pfree)$ & --\\
      $\MC(\pbounce)$ & --\\
      \hline
   \end{tabular}
  \caption{Partial dictionary
  for comparison with the non-rigorous BKL literature.
  The correspondence gives the closest analogue.
  Beware that the Kasner solution is a well-known exact solution to the vacuum Einstein
  equations, whereas the generalized Kasner spacetime is 
  a vague concept from the BKL literature.
  An analogy may be made with polydifferential operators
  with the Gerstenhaber bracket.
  The Kasner solution is analogous to the constant coefficient
  Moyal product,
  whereas the generalized Kasner spacetime is analogous
  to a generalized non-constant coefficient Moyal product,
  which is an MC-element in the associated graded of a suitable filtration
  on polydifferential operators.
  } \label{table:dict}
\end{center}
\end{table}}

\section{Introduction}

We formulate and 
prove theorems about the Belinskii-Khalatnikov-Lifshitz (BKL) proposal
for spatially inhomogeneous spacetimes in general relativity \cite{bkl}.
In this paper and a sequel, we construct formal power series
solutions in external parameters, giving meaning to `one BKL bounce'.
Naively a problem with nonzero homological obstruction space,
a filtration is used to set up an 
unobstructed Maurer-Cartan (MC) perturbation problem.
See Gerstenhaber \cite{Gerstenhaber}.

In {\citeone} the vacuum Einstein equations were formulated
as the MC-equation $[x,x]=0$ for an unknown $x \in \E^1$ of degree one
in a graded Lie algebra (gLa) $\E$.
Here we construct formal power series solutions,
namely MC-elements in the gLa $\E[[s]]$
where $s$ stands for one or more symbols.
These are formal power series whose leading terms are,
in BKL terminology that we will not adopt, the so-called
\begin{itemize}
  \item generalized Kasner spacetime
  \item one bounce spacetime
\end{itemize}
They are rigorous
building blocks for the BKL proposal. 
The full BKL proposal would involve sticking together an infinite sequence of
bounces.

This paper is organized around a gLa filtration $(F_p \E)_{p \geq 0}$
that we call the BKL filtration.
It provides a clear framework for the BKL proposal.
See also Table \ref{table:dict}.
The BKL filtration is in all homological degrees, comprehensively organizing
the gauge symmetries, the unknown, the MC-equations,
and the differential identities.

Recall from {\citeone}
that $\E$ is a graded Lie algebroid over the graded commutative algebra $\ww$;
here $W$ is the $\RR$-module of sections of
a real vector bundle with a conformal inner product
of signature ${-}{+}{+}{+}$ on a 4-dim manifold.
So $\E$ comes with an anchor map $\E \to \Der(\ww)$.
The BKL filtration is defined once an orthogonal decomposition of $W$
into rank two timelike and spacelike submodules is chosen.

\begin{theorem}[The 1-index BKL filtration]
Let $W = W_0 \oplus W_1$ be a fiberwise orthogonal decomposition
  with $W_0$, $W_1$ having signatures ${-}{+}$ and ${+}{+}$ respectively.
  Associated to it is a canonical gLa filtration of $\E$,
  induced by a corresponding $\Z_2$-graded filtration
  on $\ww$ via the anchor map. 
  Explicitly,
  with the summands given in Table \ref{table:EG},
  \[
    F_p \E = \textstyle \bigoplus_{q \leq p}
    \bigoplus_{p_2,p_3} G_{q,p_2,p_3} \E
  \]
  In the table, it is understood that $W_0 = \RR \theta_0 \oplus \RR \theta_1$
  and $W_1 = \RR \theta_2 \oplus \RR \theta_3$
  for a conformally orthonormal basis $\theta_0,\theta_1,\theta_2,\theta_3$.
\end{theorem}
To obtain the 2- and 3-index BKL filtrations,
choose a conformally orthonormal basis associated to which are
three orthogonal decompositions of $W$,
hence three 1-index BKL filtrations denoted $F$, $F'$, $F''$ respectively,
see \defref{23ind}. Their intersections 
\begin{equation}\label{eq:f2f3}
F_{p_2p_3} \E = F_{p_2}' \E \cap F_{p_3}'' \E
\qquad
F_{p_1p_2p_3} \E = F_{p_1} \E \cap F_{p_2}' \E \cap F_{p_3}'' \E
\end{equation}
are the 2- and 3-index filtrations.
The 3-index grading $G_{p_1p_2p_3} \E$ gives all these filtrations.
The significance of these filtrations is that they guide
one in constructing MC-elements for BKL.

We motivate these filtrations
by reference to the existing non-rigorous BKL literature,
but also independently on algebraic and geometric grounds.
Once these filtrations are accepted,
there is a logical framework in which to prove theorems.
In this introduction
we pretend that $p$ is a single index,
rather than two or three.

\tabdict

The filtration organizes formal power series expansions.
Namely, MC-elements are sought in the Rees algebra
\[
  \px \;=\; \{ \textstyle \sum_{p \geq 0} s^p x_{(p)} \mid x_{(p)} \in F_p \E\}
\]
morally the generating function of the filtration.
We discuss such `filtered expansions' generally,
for any gLa, in \secref{fe}.
So the $p$-th term of a formal power series MC-element
is drawn from $F_p\E^1$.
The clumsy notation $x_{(p)}$
indicates that constructing the $x_{(p)}$ in succession
is conceptually the wrong approach.
Instead, one solves the equations successively
modulo the descending chain of ideals
\[
  s\px \supset s^2 \px \supset s^3 \px \supset \ldots
\]

The leading term of such an expansion
is naturally an MC-element $x_0 \in \MC(\ax)$ of the associated graded gLa $\ax=
\Gr \px = \px/s\px$.
It contains more data than the naive leading term $x_{(0)}$.
An important point is that
the associated differential $[x_0,-]$
on $\ax$ controls perturbation theory on $\px$,
in the usual sense of MC perturbation and obstruction theory.
So its first homology is the linearized solution space,
its second homology is the obstruction space.

Even if the differential $[x_{(0)},-]$ on $\E$ has a nonzero obstruction space,
the differential $[x_0,-]$ on $\ax$ may have a vanishing obstruction space.
This scenario plays out for BKL, as we show in examples,
giving an unobstructed problem.

An informal motivation for the BKL filtration is in \secref{mot}.
The 1-index BKL filtration is in \secref{fil1}.
The 2- and 3-index filtrations,
whose Rees algebras we suggestively call $\pbounce$ and $\pfree$,
are in \secref{fil23}.
We construct MC-elements
of the associated gradeds $\abounce$ and $\afree$,
that are interesting closed-form solutions,
in Sections \ref{sec:mcafree} and \ref{sec:bouncesol}.
This requires solving the so-called constraint equations,
which we do in the real analytic class
for simplicity, in \secref{constraints}.
We discuss pertinent gLa automorphisms in \secref{generality},
and review spectral sequences in \appendixref{ssfc}.

The differential on the associated gradeds defines a filtered complex that
is studied using two tools:
spectral sequences for a filtered complex;
and a contraction based on gauge-fixing {\citeone} for general relativity.

Non-rigorous work on
the BKL proposal \cite{bkl} include \cite{berger,lim,uggla}
and references therein.
See also the exchange \cite{exc,exc2}.
Table \ref{table:dict} provides a partial dictionary
for comparison with the existing BKL literature.

The BKL proposal has so far only been rigorously implemented 
for the non-generic class of spatially homogeneous spacetimes,
see \cite{bklhom} and references therein,
where the vacuum Einstein equations reduce to
ordinary differential equations.
There one-bounce solution
are constructed via a fixed point method,
and an infinite sequence of bounces
were stuck together to get a semiglobal solution.
The spatially homogeneous problem is non-trivial since
small divisors appear when sticking bounces together.
But the gauge freedom, that is the homological nature
of general relativity, can be avoided,
unlike in the spatially inhomogeneous case
that this paper is about.

This work departs conceptually from the existing, non-rigorous BKL literature
by consistently employing a homological framework.
This does due justice to the gauge
theory nature of general relativity.
An attempt to construct formal power series can run into obstructions,
and the homological framework treats the obstruction space as an entity.
The gLa $\E$ contains geometrically degenerate elements,
without metric counterpart, as regular elements,
including elements with a frame collapsed to timelike lines,
which we perturb to construct BKL bounces,
in a controlled way dictated by the BKL filtration.
 
Combining this paper
with a no obstructions result in the companion paper {\citethree},
one obtains the existence of 
formal power series solutions to the vacuum Einstein equations,
interpreted as one BKL bounce.

 \begin{nonumtheorem}[One bounce MC-elements - informal version]
   Spatially homogeneous elements in $\MC(\abounce)$
   have zero homological obstruction space $H^2$ in $\afree$ and $\abounce$.
   This yields MC-elements in $\pfree$ and $\pbounce$ respectively.
 \end{nonumtheorem}
 This theorem concerns spatially 
 inhomogeneous perturbations of homogeneous MC-elements;
 inhomogeneous because the spaces $\afree$ and $\abounce$
 involve no homogeneity condition.
 Starting from homogeneous MC-elements is for simplicity.

We speculate that matching adjacent
BKL bounces at the formal perturbative level will define
a discrete dynamical system of $L_\infty$ maps,
describing gravitational scattering across a BKL bounce.
In this paper we construct what one would expect to be the leading
term of this scattering map.
\begin{nonumtheorem}[Bounce map - informal version of \theoremref{ibm}]
  A subset of $\MC(\abounce)$
  has future and past asymptotic limits
  (see \remarkref{xsw})
  that are in $\MC(\afree)$.
  This gives
  a partially defined map $\mathcal{B}: \mcn \nrightarrow \mcn$
  where `normal' means that elements
  are in a distinguished coordinate system and frame.
 \end{nonumtheorem}
 Iterating $\mathcal{B}$ gives a discrete dynamical system.
 Analyzing this system can give insight
 into the problem of concatenating several bounces, but
 it does not amount to studying the
 full inhomogeneous BKL proposal, because it is only for the leading term.

The following quote from BKL \cite{exc2} is prescient:
\begin{quote}
  {\footnotesize
    As to the general solution 
    (the oscillatory regime), it cannot be represented
    in a closed analytical form, although it admits a very
    detailed description \cite{refexc2}
    (the same refers even to the simpler case of
    the oscillatory regime in homogeneous models).
    The construction of this solution
    [...] is based only on [...] the estimation of the terms
    in the equations which are omitted in the asymptotic limit.
    It is just these estimates [...] which prove the possibility
    to neglect the particular terms with spatial derivatives in the equations.}
\end{quote}
It shows that BKL viewed their
calculations as providing an approximation to true solutions,
and that estimates are called for.
We do not see that their estimates are amenable to rigorous mathematics.
In this paper we elevate their `detailed description'
to a formal power series expansion for one bounce.
Natural next steps are:
\begin{itemize}
  \item The construction of true one bounce solutions,
either by showing that the formal series converge,
or using energy estimates
to construct nearby solutions.
Energy estimates for the MC-equation in the gLa $\E$
are obtained by gauge-fixing.
By construction,
the gauges in {\citeone} yield
symmetric hyperbolic equations.
\item By matching adjacent BKL bounces,
  first at the formal perturbative level.
  This will require suitable gauge transformations,
  and as we have mentioned,
  ought to result in a discrete dynamical system of $L_\infty$ maps.
\end{itemize}

The BKL claim that one can concatenate infinitely many bounces,
though made rigorous for spatially homogeneous spacetimes \cite{bklhom},
remains extraordinary for spatially inhomogeneous spacetimes.
We believe that the
validity of this claim
will largely be decided at the formal perturbative level.
Many potential issues have not been dealt with by BKL,
including the global causal structure; small divisors;
the necessity of repeated gauge transformations;
and repeated localization in space to avoid resonances.
The spatially inhomogeneous BKL proposal remains wide open.


\section{Filtered expansions}\label{sec:fe}

Let $\gx$ be a real gLa $\gx$ with a filtration.
The filtration will be used as a comprehensive organizing
tool for formal expansions.
In favorable cases, a problem with nonzero naive obstruction space
is turned into one with zero obstruction space.

Define the Maurer-Cartan set
\[
\MC(\gx) = \{x \in \gx^1 \mid [x,x] = 0\}
\]
which would be a real quadratic variety
if $\gx$ were finite-dimensional.
In favorable cases, the infinitesimal action
of the Lie algebra $\gx^0$ can be integrated,
tracing out orbits on $\MC(\gx)$,
and it is the space of orbits that is the basic object.
Rigorous sense can be made in formal perturbation theory
about a point of $\MC(\gx)$,
effectively by tensoring with an auxiliary nilpotent algebra,
giving the deformation functor.
\begin{definition}[Filtration, its Rees algebra and its associated graded]\label{def:fi}
  Let $\gx$ be a real gLa.
  By a filtration we mean a collection $(F_p\gx)_{p \geq 0}$ where
  \begin{itemize}
    \item $F_p \gx \subset \gx$ is a graded subspace
    \item $F_p \gx \subset F_{p+1}\gx$
    \item there is a finite $p$ such that $F_p \gx = \gx$
    \item $F_p\gx$ is consistent with the bracket,
      \begin{equation}\label{eq:cbr}
        [F_p \gx,F_q \gx] \subset F_{p+q}\gx
      \end{equation}
  \end{itemize}
  Associated to the filtration is its Rees algebra.
  Using a formal parameter $s$,
  the Rees algebra is the following gLa over $\R[[s]]$,
  a subalgebra of $\gx[[s]]$:
  \[
    \pxg = \{ \textstyle\sum_{p \geq 0} s^p x_{(p)} \mid x_{(p)} \in F_p \gx\}
  \]
  The associated graded is the graded gLa $\axg = \pxg / s\pxg$.
\end{definition}
Condition \eqref{eq:cbr}
encodes a wealth of conditions.
In particular,
for the Lie algebra representation $\gx^0 \otimes \gx^i \to \gx^i$,
and for the MC-map $\gx^1 \otimes \gx^1 \to \gx^2$.

Assuming each $F_p\gx$ admits a complement in $F_{p+1}\gx$,
then non-canonically $\axg \simeq \gx$
as graded vector spaces
and $\pxg \simeq \gx[[s]]$ as graded $\R[[s]]$-modules,
however not as gLa.

\begin{definition}[Filtered expansion, or expansion subordinate to a filtration]
  \label{def:fexp}
For every $x \in \MC(\pxg)$
we say that its leading term is the element $x \bmod s\pxg \in \MC(\axg)$.
Conversely, to every $x_0 \in \MC(\axg)$
we associate the moduli space of solutions
\begin{equation}\label{eq:msa}
  \frac{\{
  x \in \MC(\pxg) \mid x \bmod s \pxg = x_0
  \}}{\exp(s \pxg^0[[s]])}
\end{equation}
where the denominator is a Lie algebra over $\R[[s]]$
using Baker-Campbell-Hausdorff.
\end{definition}
\begin{remark} \label{remark:pli}
  If $x = \sum_p s^p x_{(p)} \in \MC(\pxg)$
then we expressly do not view
$x_{(0)} \in \MC(F_0\gx)$ as the leading term,
and we contend that constructing the
$x_{(p)}$ in succession is the wrong approach.
Instead, formal perturbation theory
acquires a uniform structure if one solves the equations
successively modulo the descending chain of ideals
\begin{equation}\label{eq:dci}
      s \pxg \supset s^2 \pxg \supset s^3\pxg  \supset \ldots
    \end{equation}
That is, we view $\pxg$ as the projective limit
$\varprojlim_p \pxg / s^p \pxg$.
\end{remark}
\begin{lemma}[The differential controlling formal perturbations]
  Given an element $x_0 \in \MC(\axg)$,
  then formal perturbations about $x_0$ in the sense of \eqref{eq:msa},
  organized using the
    descending chain of ideals \eqref{eq:dci},
    are controlled by the differential
    \begin{equation}\label{eq:da}
      d = [x_0,-] \in \End^1(\axg)
    \end{equation}
    In particular $H^2(d)$ is the obstruction space
    in the sense of Gerstenhaber.
\end{lemma}
\begin{proof}
  This is a slightly informal 
  version of a theorem proved in {\citeone}.
\qed\end{proof}

So filtered expansions in $\pxg$ are based on some element of $\MC(\axg)$
and are controlled by the associated differential on $\axg$.
This differential is lower triangular.
\begin{lemma}[Block lower triangularity]
  Relative to the grading $\axg = \bigoplus_p F_p \gx/ F_{p-1} \gx$,
  the differential \eqref{eq:da} is a block lower triangular,
  finite square matrix.
\end{lemma}
\begin{proof}
 The $p$-th subdiagonal is determined by $x_{(p)} \bmod F_{p-1}\gx$.
\qed\end{proof}
So $d$ is naturally studied via its spectral sequence,
reviewed in \appendixref{ssfc}.
A particularly intuitive situation occurs if the block diagonal part of
$d$ has nonzero second homology,
but the subdiagonal entries kill off that homology
to give an unobstructed problem.
Situations of this kind motivate the abstract setup.

Similar statements hold for 2-index and many-index filtrations.

\begin{example}[Generalized Moyal product analogous to generalized Kasner spacetime]
  We will make sense of
  BKL's generalized Kasner spacetime as an MC-element
  in an associated graded gLa.
  An analogous algebraic example is
  a generalized Moyal product.
  Let $C^\infty$ be the smooth functions on $\R^n$.
  Let $\gx$ be the gLa of
  polyderivations \cite{k} with the Gerstenhaber bracket,
  so $\gx^i \subset \Hom_\R ((C^\infty)^{\otimes (i+1)},C^\infty)$.
  The elements of $\MC(\gx)$ are the associative maps,
  for example pointwise multiplication.
  Let $F_p \gx$ be the filtration by the total number of derivatives;
  $F_p\gx \neq \gx$ for all $p$
  but we gloss over that.
  Given
  $\pi^{ik} = -\pi^{ki} \in C^\infty$
  the generalized non-constant coefficient Moyal product
  \[
    \textstyle(f,g) \mapsto \sum_p \frac{1}{p!}
  s^{2p}
  \pi^{i_1k_1} \cdots \pi^{i_pk_p}
  (\p_{i_1}\cdots \p_{i_p} f)(\p_{k_1}\cdots \p_{k_p} g)
\]
  is not in $\MC(\pxg)$
but it is in $\MC(\axg)$.
\end{example}

\section{The BKL filtration, informal motivation} \label{sec:mot}

There are different motivations for the BKL filtration:
\begin{itemize}
  \item \emph{Historic motivation:
    The filtration distills the algebraic content of the BKL proposal.}
    The BKL proposal simplifies
    the Einstein equations in an ad-hoc way,
    in particular by dropping terms containing spatial derivatives.
    What remains are ordinary differential equations
    along distinguished timelike lines
    that foliate the spacetime.
The simplified equations are explicitly solvable.
    We make sense of the ad-hoc simplifications
    by introducing 
    a filtration of $\E$ that we call 
    the BKL filtration. The $\MC(\ax)$ equations
    amount to the simplified equations of BKL.
  \item \emph{Geometric motivation:
    The filtration encodes a controlled degenerate to nondegenerate perturbation.}
    It is natural
    to try to construct elements in $\MC(\E)$
    by perturbing simpler ones.
    Every element of $\E^1$ defines
    a map $W^\ast \to \Der(\RR)$ of rank four $\RR$-modules,
    conventionally called a frame.
    One can try to
    perturb explicitly solvable MC-elements with a degenerate frame of rank one
    into ones with a nondegenerate frame of rank four,
    and the BKL filtration is designed to do that.
  \item \emph{Algebraic motivation:
      The filtration is functorial given an orthogonal direct sum decomposition 
    of $W$.}
 One can appreciate the simplicity and functoriality of the filtration,
    and how it is constructed using the anchor map $\E \to \Der(\ww)$.
  \item \emph{The ends justify the means.}
    The BKL filtration enables us to construct
    new formal power series solutions,
    in external parameters,
    naturally interpreted as inhomogeneous one bounce spacetimes.
\end{itemize}
\section{The 1-index BKL filtration} \label{sec:fil1}

From this point on, we use notation from \citeone.
The gLa $\E = \L/\I$ is also a graded Lie algebroid
over the graded commutative algebra $\ww$.
Tensor products
are over $\RR$.
Here $W$,
as in {\citeone}, is a rank four free $\RR$-module
with a conformal inner product of signature ${-}{+}{+}{+}$.
A conformally orthonormal basis is denoted
$\theta_0,\ldots,\theta_3$ as in {\citeone}.

We define a filtration on $\E$
that is functorially associated to an
orthogonal direct sum decomposition of $W$ into
two 2-dimensional submodules, one timelike 
meaning it has signature ${-}{+}$,
one spacelike meaning it has signature ${+}{+}$.

The key is the consistency \eqref{eq:cbr} with the bracket,
a quadratic system of inequalities.
Some filtrations come from a filtered representation space,
in fact, a gLa representation
$\E \to \End X$
with a filtration of $X$ by graded subspaces
induces a gLa filtration of $\E$.
We do something like this, 
except that we use filtrations with compatible $\Z_2$-grading,
and our notion of induced filtration
also depends on the $\Z_2$-grading.

\begin{definition}[Filtration with $\Z_2$-grading] \label{def:z2f}
  Let $X$ be a $\Z$-graded vector space.
  A filtration with compatible $\Z_2$-grading is
  data $X_0$, $X_1$ and $(F_p X)_{p \geq 0}$ where:
  \begin{itemize}
    \item $X_0$, $X_1$, $F_pX$ are graded subspaces of $X$,
      with respect to its $\Z$-grading
    \item $X = X_0 \oplus X_1$ as an internal direct sum
    \item $F_pX \subset F_{p+1}X$, and $F_p X = X$ for some finite $p$
    \item $F_pX = F_{\lhd p}X \oplus F_{\lhd p-1}X$
      where $F_{\lhd p}X = F_pX \cap X_{p \bmod 2}$
  \end{itemize}
  It is equivalently determined by a collection $(F_{\lhd p} X)_{p \geq 0}$
  with suitable properties.
  Beware that the $\Z_2$-grading
  and the original $\Z$-grading are separate data.
\end{definition}
Below $\End(X)$ is the gLa
of $\Z$-graded endomorphisms with the graded commutator.
It also inherits a $\Z_2$-grading from $X$.
\begin{lemma}[Induced gLa filtration]\label{lemma:z2find}
  Suppose
  $\gx$ is a gLa that also has a $\Z_2$-grading $\gx =  \gx_0 \oplus \gx_1$
  compatible with the bracket and the $\Z$-grading.
  Suppose $X$ is as in \defref{z2f}.
  Given a gLa representation $\gx \to \End(X)$
  that respects the $\Z_2$-gradings,
  then there is on $\gx$
  an induced filtration as in  \defref{z2f} determined by
\[
    F_{\lhd p} \gx
    \;=\;
    \{ g \in \gx_{p \bmod 2}
    \mid \forall q \geq 0:\; g(F_{\lhd q}X) \subset F_{\lhd q+p}X
  \}
\]
Furthermore,
$[F_{\lhd p} \gx,F_{\lhd q} \gx] \subset F_{\lhd p+q}\gx$
so both the $\Z_2$-grading
and the induced filtration $F_p \gx$ are consistent with the bracket.
\end{lemma}
\begin{proof}
  Omitted.
  The fact that the filtration
  on $X$ only starts at $p=0$
  and exhausts at a finite $p$ is important.
\qed\end{proof}
\begin{definition}[Orthogonal decomposition of $W$] \label{def:odw}
  This is a decomposition
 \[
   W = W_0 \oplus W_1
 \]
 with
 $W_0 = \RR \theta_0 \oplus \RR \theta_1$
 and
 $W_1 = \RR \theta_2 \oplus \RR \theta_3$
 for some conformally orthonormal basis.
 Set
 $F_{\lhd 0} W = W_0$ and $F_{\lhd 1} W = W_1$
 to obtain a filtration as in \defref{z2f}.
\end{definition}

One could define a groupoid of rank 4 vector bundles
with both a conformal inner product and an orthogonal decomposition
of $W$. The BKL filtration defines
a functor out of this groupoid, into,
for example, the category of filtered gLa.

\begin{lemma}[The 1-index BKL filtration] \label{lemma:fil1}
  For every orthogonal decomposition $W = W_0 \oplus W_1$
  there are associated
  filtrations as in \defref{z2f} as follows:
  \begin{itemize}
    \item
  On $\ww$
  the smallest such filtration
  that respects the unital algebra structure
  and the injection $W \inj \ww$. So
  $F_{\lhd p}(\ww)$ is $\RR$-spanned by
  any monomial in $\theta_0,\theta_1,\theta_2,\theta_3$,
  such that the total degree in $\theta_2,\theta_3$
  is at most $p$ and is
  even respectively odd depending on whether $p$ is even respectively odd.
  So $F_2(\ww) = \ww$.
\item On $\ww \otimes \Omega$ the filtration
  $F_{\lhd p} (\ww \otimes \Omega) = F_{\lhd p}(\ww) \otimes \Omega$.
  Here $\Omega = \wedge^4 W$ \citeone.
\item On $\E$ the filtration induced by
  the $\Z_2$-grading $\E = \E_0 \oplus \E_1$
  coming from $W = W_0 \oplus W_1$,
  together with
  \lemmaref{z2find} via
  the natural gLa representation \citeone
  \begin{equation}\label{eq:erep}
    \E \to \End(\ww) \oplus \End(\ww \otimes \Omega)
  \end{equation}
  The representation on the first summand is
  the anchor $\E \to \Der(\ww)$.
  The representation descends from a
  representation of $\L = \ww \otimes \CDerEnd(W)$,
  given explicitly by
  $\omega \delta \mapsto (\omega'z
  \mapsto (\omega \delta(\omega'))z + (\omega \omega') \delta(z))$
  with $z=1$ respectively $z \in \Omega$.
  \end{itemize}
  We call $(F_p\E)_{p \geq 0}$ the BKL filtration
  associated to the decomposition $W = W_0 \oplus W_1$.
\end{lemma}
\begin{proof}
  The $\Z_2$-gradings are natural and are compatible.
\qed\end{proof}
One can extend the filtration to a graded Lie algebroid filtration of $\E$
  over the filtered graded commutative algebra $\ww$,
with filtered module and anchor maps.

  We have defined the BKL filtration via a representation,
  which makes it look natural.
  One would expect that this representation
  should play an important role in using the BKL filtration.
  Actually, we proceed with a more explicit description.
  
 \begin{lemma}[Decompositions of $\E$]
   For every conformally orthonormal basis of $W$,
   there are graded free $\RR$-submodules\footnote{%
     We are abusing notation of {\citeone}.
     Here $\E_G$ is a new symbol, not immediately related
     to the output of the gauge fixing algorithm in {\citeone},
     though it can be obtained from it by a limiting procedure.
   }
   $G_{\alpha} \E_G \subset \E$ defined by Table \ref{table:EG},
   with $\alpha = p_1p_2p_3$.
   The ranks are given in the table.
   Set
   \begin{align*}
      \E_G
      &\;=\; \textstyle\bigoplus_\alpha G_\alpha \E_G\\
      G_\alpha\E
      &\;=\; G_{\alpha}\E_G \oplus \theta_0 (G_\alpha\E_G)
   \intertext{%
     These are actual internal direct sums in $\E$.
    Furthermore:}
     \E & \;=\; \E_G \oplus \theta_0 \E_G\\
      \E & \;=\; \textstyle\bigoplus_\alpha G_\alpha \E
   \end{align*}
   Here $\ww$-module multiplication by $\theta_0$ acts injectively.
   All direct summands are free graded $\RR$-submodules.
   The $\Z_{\geq 0}^3$-grading $G_{\alpha}\E$
   is invariant under scalings of the conformally orthonormal basis by a positive function.
\end{lemma}
\newcommand{\spx}{\,\,\,}
 \begin{table}
   \footnotesize{
   \[
   \begin{array}{|l|l|l|}
     \hline
     \alpha = p_1p_2p_3 & \text{$G_\alpha \E_G \subset \E$ is the $\RR$-span of these elements} & \text{$\RR$-rank} \\
     & 
     \text{note that
     $G_\alpha \E = G_\alpha \E_G \oplus \theta_0(G_\alpha \E_G)$
     }
     &\\
     \hline
     \hline
     000 & \Der(\RR),\spx\sigma_0,\spx\theta_0 \sigma_0 + \theta_1\sigma_1,\spx
             \theta_0\sigma_0 + \theta_2\sigma_2,\spx \theta_0\sigma_0 + \theta_3\sigma_3, & 9\\
             & 
             \theta_2\theta_3\sigma_{23}+\theta_3\theta_1\sigma_{31}+\theta_1\theta_2\sigma_{12}
             + 2\theta_0\theta_1\sigma_1+2\theta_0\theta_2\sigma_2+2\theta_0\theta_3\sigma_3 & \\
     \hline
     200 & -\theta_1\sigma_{23} + \theta_2\sigma_{31} + \theta_3\sigma_{12} & 1\\
     \hline
     020 & +\theta_1\sigma_{23} - \theta_2\sigma_{31} + \theta_3\sigma_{12} & 1\\     
     \hline
     002 & +\theta_1\sigma_{23} + \theta_2\sigma_{31} - \theta_3\sigma_{12} & 1\\
     \hline
     011 & \sigma_1,\spx
             \sigma_{23},\spx
             \theta_1 \Der(\RR),\spx
             \theta_0\sigma_1+\theta_1\sigma_0,\spx
             \theta_2\sigma_3+\theta_3\sigma_2, & 11\\
           & \theta_3\sigma_{31},\spx
             \theta_2\sigma_{12},\spx
             \theta_0\theta_2\sigma_{12}+\theta_1\theta_2\sigma_2 & \\
             \hline
     101 & \sigma_2,\spx
             \sigma_{31},\spx
             \theta_2 \Der(\RR),\spx
             \theta_0\sigma_2+\theta_2\sigma_0,\spx
             \theta_3\sigma_1+\theta_1\sigma_3, & 11\\
           & \theta_1\sigma_{12},\spx
             \theta_3\sigma_{23},\spx
             \theta_0\theta_3\sigma_{23}+\theta_2\theta_3\sigma_3 & \\
             \hline
     110 & \sigma_3,\spx
             \sigma_{12},\spx
             \theta_3 \Der(\RR),\spx
             \theta_0\sigma_3+\theta_3\sigma_0,\spx
             \theta_1\sigma_2+\theta_2\sigma_1, & 11\\
           & \theta_2\sigma_{23},\spx
             \theta_1\sigma_{31},\spx
             \theta_0\theta_1\sigma_{31}+\theta_3\theta_1\sigma_1 & \\
             \hline
     211 & \theta_2\sigma_3 - \theta_3\sigma_2,\spx
             \theta_2\theta_3 \Der(\RR), & 7\\
           & \theta_0\theta_2\sigma_3-\theta_0\theta_3\sigma_2-2\theta_2\theta_3\sigma_0,\spx
             \theta_0\theta_2\sigma_3+\theta_0\theta_3\sigma_2-2\theta_1\theta_2\sigma_{31} & \\
             \hline
     121 & \theta_3\sigma_1 - \theta_1\sigma_3,\spx
             \theta_3\theta_1 \Der(\RR), & 7\\
           & \theta_0\theta_3\sigma_1-\theta_0\theta_1\sigma_3-2\theta_3\theta_1\sigma_0,\spx
             \theta_0\theta_3\sigma_1+\theta_0\theta_1\sigma_3-2\theta_2\theta_3\sigma_{12} & \\
             \hline
     112 & \theta_1\sigma_2 - \theta_2\sigma_1,\spx
             \theta_1\theta_2 \Der(\RR), & 7\\
           & \theta_0\theta_1\sigma_2-\theta_0\theta_2\sigma_1-2\theta_1\theta_2\sigma_0,\spx
             \theta_0\theta_1\sigma_2+\theta_0\theta_2\sigma_1-2\theta_3\theta_1\sigma_{23} & \\
             \hline
     222 & \theta_0\theta_1\sigma_{23}+\theta_0\theta_2\sigma_{31}
             + \theta_0\theta_3\sigma_{12}
             - 2 \theta_2\theta_3\sigma_1
             - 2 \theta_3\theta_1\sigma_2
             - 2 \theta_1\theta_2\sigma_3, & 6 \\
           & \theta_1\theta_2\theta_3 \Der(\RR), & \\
           & \theta_0\theta_2\theta_3\sigma_1
            +\theta_0\theta_3\theta_1\sigma_2
            +\theta_0\theta_1\theta_2\sigma_3
            +3\theta_1\theta_2\theta_3\sigma_0 &
     \\
     \hline
     \text{else} & \text{none} & 0\\
     \hline
 \end{array}
 \]}
 \caption{
 Definition of $G_\alpha\E_G$
 associated to a 
 conformally orthonormal basis $\theta_0,\theta_1,\theta_2,\theta_3$.
 We omit wedges, so $\theta_0\theta_1 = \theta_0 \wedge \theta_1$.
 The basis-dependent injection $\Der(\RR) \inj \CDerEnd(W)$ is implicit \citeone.
 All elements in the table are
 elements of $\E$ via the canonical surjection $\mathcal{L} \surj \E$.
 }\label{table:EG}
 \end{table}

\begin{proof}
  The direct sums are well-defined;
  linear independence is self-evident in $\L$ because it is a product,
  but not self-evident in $\E$.
  The $G_{\alpha}\E_G$
  are not invariant under conformal scalings of the basis,
  since the injection $\Der(\RR) \inj \CDerEnd(W)$
  changes by terms proportional to $\sigma_0$.
  Hence in $G_{211}\E_G$
  the term $\theta_2\theta_3\Der(\RR)$
  generates new terms $\theta_2\theta_3\sigma_0$;
  more details will be in \theoremref{g13}.
  While such terms are not contained in $G_{211}\E_G$,
  they are in $G_{211}\E$. This example generalizes.
\qed\end{proof}
The space $\E_G$
will be used for calculations,
at the moment it plays a minor role.
\begin{lemma}[The 1-index BKL filtration, 
       explicit definition] \label{lemma:fil1next}
  The 1-index BKL filtration associated to
  $W_0 = \RR \theta_0 \oplus \RR \theta_1$
  and $W_1 = \RR \theta_2 \oplus \RR \theta_3$ is
  \[
    F_{\lhd p} \E = \textstyle \bigoplus_{k \geq 0}
    \bigoplus_{p_2,p_3} G_{p-2k,p_2,p_3} \E
  \]
\end{lemma}
\begin{proof}
  The $\Z_2$-grading is important because,
  for example, $\theta_0\theta_3\sigma_1+\theta_0\theta_1\sigma_3
  \in G_{101}\E$
  is mapped to the zero map under \eqref{eq:erep}
  and hence would be a candidate for $F_0 \E$
  if it were not for the $\Z_2$-grading, which puts it in $F_{\lhd 1}\E$.
  One checks that in Table \ref{table:EG},
  the entries with $p_1$ even (odd) are in the even (odd)
  sector of $\E = \E_0 \oplus \E_1$
  for its $\Z_2$-grading associated to $W_0 \oplus W_1$.
  For example, $A_{\pm} = \theta_2 \sigma_3 \pm \theta_3\sigma_2$
  are even and indeed appear in the table with $p_1$ even.
  Note that $A_+$ annihilates $\theta_0$, $\theta_1$
  and from $\theta_2$, $\theta_3$
  produces $\theta_3\theta_0$, $\theta_2 \theta_0$ terms
  not increasing filtration degree, and using the derivation
  property we get $A_+ \in F_{\lhd 0} \E$.
  By contrast, applying $A_-$ to $\theta_0 \in F_{\lhd 0}(\ww)$
  gives $\theta_2\theta_3 \notin F_{\lhd 0}(\ww)$,
  which implies $A_- \notin F_{\lhd 0} \E$,
  so it is only in $F_{\lhd 2} \E$.
  By similar arguments, no linear combination of elements in
  $G_{2p_2p_3} \E_G$, and in fact $G_{2p_2p_3}\E$,
  can be in $F_{\lhd 0}\E$.
  The $\ww \otimes \Omega$
  part in \eqref{eq:erep} is used,
  for example, to make
  $\theta_0\theta_1\theta_2\theta_3 \sigma_0 \notin F_{\lhd 0}\E$.
\qed\end{proof}
As an application, note that
elements of $F_0 \E^1$ define a frame of rank at most two,
hence a degenerate frame, because
it contains
$\theta_i \Der(\RR)$
with $i=0,1$ but not $i=2,3$.
Degenerate frames were part of the motivation in \secref{mot}.

\section{The 2-index and 3-index BKL filtrations} \label{sec:fil23}

\begin{definition}[Full orthogonal decomposition of $W$] \label{def:fud}
  This is a decomposition
 \[
   W = \RR \theta_0 \oplus \RR \theta_1 \oplus \RR \theta_2 \oplus \RR \theta_3
 \]
 for some conformally orthonormal basis,
 with the understanding that two
 bases that differ only by a scaling by a positive function define
 the same decomposition.
\end{definition}
\begin{definition}[2-index and 3-index BKL filtrations] \label{def:23ind}
Given an orthogonal decomposition of $W$ introduce the following 1-index BKL filtrations:
\begin{itemize}
\item
$(F_p\E)$ the BKL filtration
  associated to
  $W_0 = \RR \theta_0 \oplus \RR \theta_1$,
  $W_1 = \RR \theta_2 \oplus \RR \theta_3$.
\item
$(F_p'\E)$ the BKL filtration
  associated to
  $W_0 = \RR \theta_0 \oplus \RR \theta_2$,
  $W_1 = \RR \theta_3 \oplus \RR \theta_1$.
\item
$(F_p''\E)$ the BKL filtration
  associated to
  $W_0 = \RR \theta_0 \oplus \RR \theta_3$,
  $W_1 = \RR \theta_1 \oplus \RR \theta_2$.
\end{itemize}
    Then \eqref{eq:f2f3} defines the 2-index and 3-index BKL filtrations.
    They are gLa filtrations, and\footnote{Here $\beta = q_1q_2q_3 \leq \alpha = p_1p_2p_3$
    iff $q_1 \leq p_1$
    and $q_2 \leq p_2$ and $q_3 \leq p_3$}
    $F_{\alpha} \E = \bigoplus_{\beta \leq \alpha}
    G_{\beta} \E$.
    The filtration
    is free over $\RR \mathbbm{1} \oplus \RR \theta_0$ with $\theta_0^2 = 0$,
    the exterior algebra in one generator.
  \end{definition}
  Elements
  of $F_{00}\E^1$ and $F_{000} \E^1$
  define a frame of rank at most one,
  hence a degenerate frame, because
they contain
$\theta_i \Der(\RR)$
only when $i=0$.
See also \secref{mot}.

  \begin{lemma}[Rees algebras]
Let $\R[[s_1,s_2,s_3]]$
be the $\R$-algebra of formal power series in the symbols $s_1,s_2,s_3$.
    There are gLa
    \begin{align*}
      \pbounce & \;=\;
      \{ \textstyle\sum_{p_2,p_3} s^{p_2}_2 s^{p_3}_3 x_{p_2p_3}
      \mid x_{p_2p_3} \in F_{p_2p_3}\E\}\\
      \pfree & \;=\;
      \{ \textstyle\sum_{p_1,p_2,p_3}
      s^{p_1}_1 s^{p_2}_2 s^{p_3}_3 x_{p_1p_2p_3}
      \mid x_{p_1p_2p_3} \in F_{p_1p_2p_3}\E\}
    \intertext{%
      free over $\R[[s_2,s_3]]$ and
      $\R[[s_1,s_2,s_3]]$,
  subalgebras of $\E[[s_2,s_3]]$ and $\E[[s_1,s_2,s_3]]$ respectively.
  The associated graded spaces are the real gLa}
     \abounce &\;=\; \pbounce/(s_2,s_3)\\
     \afree &\;=\; \pfree/(s_1,s_2,s_3)
\end{align*}
where $(s_2,s_3)$ is the ideal generated by these elements,
similar for $(s_1,s_2,s_3)$.
  \end{lemma}
  \begin{proof}
  Omitted.
Freeness is defined in the discussion of MC perturbations in {\citeone}.
  \qed\end{proof}

\section{Maurer-Cartan elements in $\afree$} \label{sec:mcafree}

Elements of $\MC(\afree)$
and their differentials are key in filtered expansions
on $\pfree$, see \secref{fe}.
This section contains logically complete statements
for $\MC(\afree)$.
We keep this section compact by postponing some important things:
\begin{itemize}
  \item The solution of certain partial differential equations in three dimensions,
called `constraints' in the vernacular of general relativity,
is postponed to \secref{constraints}.
  \item We write down a partial map
    from $\MC(\afree)$ to itself that we call the bounce map.
    It is astonishing that the constraints are preserved.
    This can be checked by direct, unrevealing calculation,
    but the rationale is that this map is a `scattering map'
    for an element of $\MC(\abounce)$ that we postpone to \secref{bouncesol}.
  \item The elements of $\MC(\afree)$ in this section are essentially general,
    but to see this one needs certain
    automorphisms of $\afree$ that are postponed to \secref{generality}.
  \item Studying the differential on $\afree$ is postponed.
    In particular, the obstruction space.
\end{itemize}

\begin{definition}[Foliated base manifold
                   and orthogonal decomposition] \label{def:fbm}
  The base manifold is $\R \times \spatialm$
  for a manifold $\spatialm \simeq \R^3$.
  The projection out of the first factor $t: \R \times \spatialm \to \R$
  is called `time', the partial derivative along the first factor is denoted
  \[
      D_0 = \tfrac{\p}{\p t}
  \]
  We assume 
  $W = \RR \theta_0 \oplus \RR \theta_1 \oplus \RR \theta_2 \oplus \RR \theta_3$
  is a full orthonormal decomposition,
  see \defref{fud}.
  There is canonical injection
  $\Der(\RR) \inj \CDerEnd(W)$
  that produces elements that annihilate the basis elements $\theta_i$,
  and that we use implicitly.
\end{definition}
The base manifold carries a foliation by affine lines,
suitable for BKL.
This is made precise by letting
\defref{fbm} define the objects of a groupoid,
see \theoremref{g13}.

Cyclic $(i,j,k)$
means that they run over $\{(1,2,3),(2,3,1),(3,1,2)\}$.
  \begin{lemma}[Master space
    and overparametrization of an affine gauge subspace] \label{lemma:overp}
    Given \defref{fbm} and given three symbols $s_1,s_2,s_3$, and abbreviating
    $\alpha = p_1p_2p_3$ and $s^{\alpha} = s_1^{p_1} s_2^{p_2} s_3^{p_3}$,
    define the external direct sums
      \begin{align*}
      \mathcal{U}\;& =\; \textstyle\bigoplus_{\alpha}
                            s^{\alpha} G_{\alpha} \E\\
      \mathcal{U}_G \;&=\; \textstyle\bigoplus_{\alpha}
                            s^{\alpha} G_{\alpha} \E_G
        \end{align*}
        Then the $\RR$-module `master space' $\mathcal{U}$
        is canonically isomorphic to each of
        $\E$, $\abounce$, $\afree$
        and enables identification between any two of them.
     In this way, $\mathcal{U}$ inherits three different
     real gLa structures.
     The `affine gauge subspace'
\[
  \theta_0 D_0 + \mathcal{U}_G^1 \subset \mathcal{U}^1
\]
admits an over-parametrization (surjection) in terms of, with $i=1,2,3$,
\begin{itemize}
  \item $\beta_i^{\mu},
    \, \gamma_i^a,\, \mu_i
    \in C^{\infty}(\R \times \spatialm, \R)$ with $\mu = 0,1,2,3$
    and
    $a=0,1,2,3,4,5,6$
  \item $D_i$ a frame of vector fields on $\spatialm$
    extended to $\Der(\RR)$ by $D_i(t)=0$
\end{itemize}
Concretely,
\begin{align*}
  \theta_0 D_0 \;+\; \tsum_{\textnormal{cyclic $(i,j,k)$}}
(\;\;
  -\,& \gamma_i^0 (\theta_i \sigma_i + \theta_0\sigma_0)\\
  -\,& \sv{i}^2 \gamma_i^1 (\theta_i \sigma_{jk} - \theta_j \sigma_{ki} - \theta_k \sigma_{ij})\\
  +\,& \sv{j}\sv{k} \theta_i \beta_i\\
  -\,& \sv{j}\sv{k} (\gamma_i^2 + \beta_i(\mu_1+\mu_2+\mu_3))(\theta_0\sigma_i + \theta_i\sigma_0)\\
  -\,& \sv{j}\sv{k} (\gamma_i^3 
                     -\beta_i(\mu_k)) \theta_k \sigma_{ki}\\
  -\,& \sv{j}\sv{k} (\gamma_i^4
                     +\beta_i(\mu_j)) \theta_j \sigma_{ij}\\
  -\,& \sv{j}\sv{k} \gamma_i^5(\theta_k\sigma_j + \theta_j\sigma_k)\\
  +\,& \sv{j}\sv{k}\sv{i}^2 \gamma_i^6 (\theta_j\sigma_k - \theta_k\sigma_j)
  )
\end{align*}
with the abbreviations
$\sv{i} = s_i e^{\mu_i}$ 
and
$\beta_i = \beta_i^0 D_0 + \beta_i^1 D_1 + \beta_i^2 D_2 + \beta_i^3 D_3$.
The notation in this lemma is used consistently in the following.
  \end{lemma}
  \begin{proof}
    Clear.
    Use Table \ref{table:EG}.
  \qed\end{proof}
The $\mu_i$ will simplify later statements;
we often require $D_0(\mu_i) = \gamma_i^0$.
An element for which $D_0,\beta_1,\beta_2,\beta_3$ is a frame
 may be thought of as infinitesimally nondegenerate.

\begin{lemma}[Synchronous frame] \label{lemma:g60}
If an element of $\theta_0 D_0 + \mathcal{U}_G^1$
is an MC-element
of any one of $\mathcal{E}$, $\abounce$, $\afree$ then we have
the implication
  \[
    \beta_1^0 = \beta_2^0 = \beta_3^0 = 0
    \qquad
    \Longrightarrow
    \qquad
    \gamma_1^6 = \gamma_2^6 = \gamma_3^6 = 0
  \]
\end{lemma}
\begin{proof}
  If the $\beta_i^0$ vanish,
  then only the $\gamma_i^6$
  appear as coefficients in front of $s_i^2 s_j s_k \theta_j \theta_k D_0$
  in the MC-equations,
  hence the $\gamma_i^6$ have to vanish.
  To check this, note that the bracket of any two of the eleven elements
  $D_0,D_i,\sigma_0,\sigma_i,\sigma_{jk}$
  does not produce a $D_0$.
\qed\end{proof}

\begin{lemma}[Necessary conditions for $\MC(\afree)$] \label{lemma:mcefree}
To be in
\[
  \MC(\afree) \;\cap\; (\theta_0 D_0 + {\afree^1}_{,G})
  \;\cap\; \{D_0(\mu_i) = \gamma_i^0 \}
\]
requires
$D_0(\beta_i^{0,1,2,3}) = D_0(\gamma_i^{0,1,2,3,4,6}) = 
D_0(e^{\mu_j+\mu_k} \gamma_i^5)=0$ for cyclic $(i,j,k)$.
\end{lemma}
\begin{proof}
  Required for the 
  $\afree^2/{\afree^2}_{,G}$ part of the MC-equations.
  The parametrization in \lemmaref{overp}
  was designed to get simple equations.
\qed\end{proof}
It should be clear from this lemma that the $\gamma_i^5$
are three odd ducks,
that potentially behave badly.
Luckily, one can get rid of them using
nilpotent automorphisms of\footnote{%
Crucially, these nilpotent automorphisms
are induced from automorphisms of $\pfree$
and hence extend to formal perturbation theory in $\pfree$.} $\afree$
that leave the affine gauge subspace invariant.
See 
\secref{generality}, particularly Remark \ref{remark:gen56}.
For the same reasons,
the additional conditions in
the next lemma do not entail a significant restriction of generality.

Below we associate to
the frame $D_i$
the `structure functions' $c_{jk}^i \in C^\infty(M,\R)$ by
\begin{equation}
\label{eq:mce3}
[D_j,D_k] = \tsum_{i=1}^3 c_{jk}^i D_i
\end{equation}
\begin{definition} \label{def:xpre}
Let $\mcp$
be the set of tuples
\[
  (\spatialm,g_1^0,g_2^0,g_3^0,D_1,D_2,D_3,\xi)
\]
with $g_1^0, g_2^0, g_3^0, \xi \in C^\infty(\spatialm,\R)$
subject to the `constraint' equations,
for $(i,j,k)$ cyclic,
\begin{align}
  \label{eq:mce1}
 0 & = g_2^0g_3^0 + g_3^0g_1^0 + g_1^0g_2^0\\
 \label{eq:mce2}
 0 & = D_i(g_j^0+g_k^0)
 + c_{ij}^j(g_i^0 - g_j^0)
 + c_{ik}^k(g_i^0 - g_k^0)
 - 2 D_i(\xi) g_i^0
\end{align}
with \eqref{eq:mce3} understood.
To every such tuple we associate $g_i^{1,2,3,4} \in
C^\infty(\spatialm,\R)$ by
\[
  g_i^1 = -\tfrac{1}{2} c_{jk}^i
  \qquad
  g_i^2 = D_i(\xi)
  \qquad
    g_i^3 = -c_{ki}^k - g_i^2
  \qquad
    g_i^4 = -c_{ij}^j + g_i^2
  \]
  for all cyclic $(i,j,k)$.
\end{definition}
\begin{lemma}[Sufficient conditions for $\MC(\afree)$] \label{lemma:mcefreenf}
An element is in
\[
  \MC(\afree) \;\cap\; (\theta_0 D_0 + {\afree^1}_{,G})
  \;\cap\;
  \{\;
    \mu_i = t \gamma_i^0,
    \;\;
  \beta_i = D_i,\;\;
  \gamma_i^{5,6} = 0
 \;\}
\]
iff $\gamma_i^{0,1,2,3,4} = g_i^{0,1,2,3,4}$
for some element of $\mcp$.
\end{lemma}
\begin{proof}
  By direct calculation. Here are some remarks.
  Note that $\beta_i = D_i$ means
  $\beta_i^0 = 0$ and $\beta_i^j = \delta_i^j$ for $i,j = 1,2,3$.
  Since $D_1,D_2,D_3$ is a frame,
  some constraint equations amount to the Jacobi identity and are dropped.
  The MC-equations also imply
      $D_j(g_k^2) - D_k(g_j^2) =
      \tsum_{p=1}^3 c_{jk}^p g_p^2$,
      which modulo the other MC-equations
      is equivalent to $g_i^2 = D_i(\xi)$
  by the Poincare lemma, since $\spatialm \simeq \R^3$ is simply connected.
\qed\end{proof}
\begin{remark}
  In \lemmaref{mcefreenf} a special class of MC-elements
  is given by requiring that $\gamma_i^{0,1}$ be constant and
  that $\gamma_i^{2,3,4,5,6} = 0$.
  These elements will be called homogeneous.
  Here only the
  structure coefficients
  $c^i_{jk}$ with cyclic $(i,j,k)$
  can be nonzero.
\end{remark}
\begin{definition}
        Let $\mcn \subset \mcp$ be the elements with
  \begin{align*}
  (g_1^0,g_2^0,g_3^0) & = \tfrac{1}{2}(1,-1-u,-1-\tfrac{1}{u})\\
  c_{23}^1 & = -2,\;\textnormal{equivalently}\;g_1^1 = 1
\end{align*}
  for some $u \in C^{\infty}(\spatialm,\R)$ with $u > 0$.
\end{definition}
\begin{lemma}[Normal form elements in $\MC(\afree)$] \label{lemma:mcpNEW}
  \begin{itemize}
\item[(a)]
Given an element of $\mcp$,
then for all $A,B \in C^{\infty}(\spatialm,\R)$ with
  $A,B>0$ and for all $\sigma \in \{-1,+1\}$ we have
  \begin{multline*}
      \big(\spatialm,\,Ag_1^0,\,Ag_2^0,\,Ag_3^0,\,
      \sigma A B^{g_2^0+g_3^0} D_1,\,
      \sigma A B^{g_3^0+g_1^0} D_2,\,
      \sigma A B^{g_1^0+g_2^0} D_3,\\
      \xi + \log A + (g_1^0+g_2^0+g_3^0) \log B \big)\;\in\; \mcp
    \end{multline*}
  \item[(b)]
    For every element of $\mcp$
    that satisfies
  $g_1^0 > 0$ and $g_2^0, g_3^0 < 0$ and $c_{23}^1 \neq 0$
  everywhere on $\spatialm$,
  there is a unique transformation as in (b) that
  yields an element of $\mcn$. This yields a partial map,
  that we call normalization,
  \[
      \nmlz\;:\; \mcp \nrightarrow \mcn
  \]
\end{itemize}
\end{lemma}
\begin{proof}
  Claim (a) by direct calculation;
  we will see later that these transformations actually
  come from automorphisms of $\afree$.
  For (b),
  equation \eqref{eq:mce1} implies that the sum of
  any two $g_i^0$ is negative,
  so we must take
  $A = 1/(2g_1^0)$
  and we get
  $u = - (g_1^0+g_2^0)/g_1^0 > 0$,
  and we must choose $B$, $\sigma$ so that
  $\sigma A B^{2 g_1^0} c_{23}^1 = -2$.
\qed\end{proof}
\begin{lemma} \label{lemma:outg}
  Given an element in $\mcn$ with
  $u \neq \tfrac{1}{2}, 2$ everywhere
   then
  \[
  \big(
  \spatialm,\;
  -\tfrac{1}{2},\;
  \tfrac{1}{2}(1-u),\;
  \tfrac{1}{2}(1-\tfrac{1}{u}),\;
  D_1  - \tfrac{2 g_3^4}{u-2} D_2 + \tfrac{2 u g_2^3}{2 u-1} D_3,\;
  D_2,\;
  D_3,\;
  \xi\big)
  \;\in\; \mcp
\]
\end{lemma}
\begin{proof}
  By direct unrevealing calculation, \eqref{eq:mce1} and \eqref{eq:mce2} hold.
  Alternatively and more rationally, note that the element
  of $\MC(\abounce)$ written out in
  \theoremref{bouncesol} below
  is asymptotic to an element of $\MC(\afree)$
  as $t \to -\infty$ and another as $t \to +\infty$,
  using the module identifications
  $\abounce \simeq \mathcal{U} \simeq \afree$.
\qed\end{proof}


\begin{theorem}[The inhomogeneous bounce map] \label{theorem:ibm}
  There is a partial map
  \[
      \mathcal{B}\;:\; \mcn \nrightarrow \mcn
  \]
  that assigns to every
  $X \in \mcn$,
  with associated
   $g_i^{0,1,2,3,4}$ and $u$, the following element
   $\mathcal{B}(X)$,
   in each case provided only that the partial map $\nmlz$ is defined:
  \begin{itemize}
    \item If $0 < u < 1$ and $u \neq \tfrac{1}{2}$ everywhere then set,
      with $\next{u} = \tfrac{1}{u}-1$,
      \begin{multline*}
        \mathcal{B}(X)
        \;\;
        = \;\;\nmlz\big(
  \spatialm,\;
  \tfrac{1}{2},\;
  -\tfrac{1}{2}(1+\next{u}),\;
  -\tfrac{1}{2}(1+\tfrac{1}{\next{u}}), \\
  \tfrac{1}{1-u}D_2,\;
  \tfrac{1}{1-u}D_3,\;
  \tfrac{1}{1-u}
  \big(D_1  - \tfrac{2 g_3^4}{u-2} D_2 + \tfrac{2 u g_2^3}{2 u-1} D_3\big),\;
  \xi + \log \tfrac{1}{1-u} \big)
      \end{multline*}
    \item If $u > 1$ and $u \neq 2$ everywhere then set,
      with $\next{u} = u-1$,
      \begin{multline*}
        \mathcal{B}(X)
        \;\;
        = \;\;\nmlz\big(
  \spatialm,\;
  \tfrac{1}{2},\;
  -\tfrac{1}{2}(1+\next{u}),\;
  -\tfrac{1}{2}(1+\tfrac{1}{\next{u}}), \\
  \tfrac{u}{u-1}D_3,\;
  \tfrac{u}{u-1}D_2,\;
  \tfrac{u}{u-1}
  \big(D_1  - \tfrac{2 g_3^4}{u-2} D_2 + \tfrac{2 u g_2^3}{2 u-1} D_3\big),\;
  \xi + \log \tfrac{u}{u-1} \big)
      \end{multline*}
  \end{itemize}
\end{theorem}
\begin{proof} By Lemmas \ref{lemma:mcpNEW}, \ref{lemma:outg}.
  Permute $1,2,3 \mapsto 3,1,2$
  respectively $1,2,3 \mapsto 3,2,1$
  so that $\frac{1}{2}(1-u) > 0$
  respectively $\frac{1}{2}(1-\frac{1}{u}) > 0$
  appear first,
   so that $\nmlz$ has a chance of being defined.
  We have applied
  a transformation in \lemmaref{mcpNEW} (a)
  to bring out the map $u \mapsto \next{u}$
  explicitly (this map is well-known from the BKL papers).
\qed\end{proof}

If $\mathcal{B}$ is not defined for some element,
it may be defined after restricting to an 
open subset of $\spatialm \simeq \R^3$.
So iterating $\mathcal{B}$
 may require restricting to subsets, again and again.

The conceptual interpretation of $\mathcal{B}$ is to come.
As remarked in the proofs of Lemmas \ref{lemma:mcpNEW} and \ref{lemma:outg},
it arises from taking past and future limits of
elements in $\MC(\abounce)$ and applying automorphisms to
bring them into a normal form.
These automorphisms can be interpreted as
selecting a distinguished coordinate system and frame.


\section{The constraint equations} \label{sec:constraints}

We parametrize the space of real analytic germs
of solutions to the constraint equations in \lemmaref{mcefreenf},
at the origin $0 \in \R^3$.
We confine ourselves to real analytic germs for simplicity.
Every representative of a germ is defined on some open ball,
so we get an actual solution
by defining $\spatialm$ to be such a ball.
This proves existence of elements in $\MC(\afree)$
and, via \theoremref{bouncesol} below, of $\MC(\abounce)$.
\begin{itemize}
  \item The coordinates on $\R^3$ are denoted $x^1,x^2,x^3$,
    the partial derivatives $\p_1,\p_2,\p_3$.
  \item We assume
    \begin{equation} \label{eq:assx1}
      D_1 = \p_1
      \qquad
      D_2^1(0) > 0
      \qquad
      D_3^1(0) > 0
    \end{equation}
where $D_i = D_i^1 \p_1 + D_i^2 \p_2 + D_i^3 \p_3$.
There is always a local real analytic
change of coordinates that maps the origin to itself and brings $D_1$, $D_2$, $D_3$
into this form.
\item We assume,
  with $a$ and $u > 0$ germs of real analytic functions,
  \begin{equation} \label{eq:assx2}
      (g_1^0,g_2^0,g_3^0) \;=\; \tfrac{1}{2} e^a (1,-1-u,-1-\tfrac{1}{u})
    \end{equation} 
\end{itemize}
\begin{lemma}[All real analytic germs of elements in $\MC(\afree)$] \label{lemma:rag}
  Let $\parsp$ be the set of all
  real analytic germs at $0 \in \R^3$ of functions
  \[
        D_i^j,\; c_{ij}^k,\; g_i^0,\; \xi 
  \]
  that satisfy \eqref{eq:mce3}, \eqref{eq:mce1}, \eqref{eq:mce2}, 
  \eqref{eq:assx1}, \eqref{eq:assx2},
  in particular the $g_i^0$ are given in terms of $(a,u)$.
Let $\datasp$ be the set of
  6 real analytic germs at $0 \in \R^3$
  and 9 real analytic germs at $0 \in \{ x \in \R^3 \mid x^1 = 0\}$,
  that by an abuse of notation we denote by:
  \begin{itemize}
    \item Germs $c_{31}^i$ and $c_{12}^i$ at $0 \in \R^3$.
    \item Germs $D_2^i$, $D_3^i$, $a$, $u$, $\xi$
      at $0 \in \{x \in \R^3 \mid x^1 = 0\}$,
      subject to $D_2^1(0),\,D_3^1(0),\,u(0)>0$.
  \end{itemize}
  Then the map $\parsp \to \datasp$,
   given  by the identity on the first 6 and by restriction
   to the $x^1 = 0$ hypersurface for the other 9,
   is bijective.
\end{lemma}
\begin{proof}
  We use Cauchy Kowalewski
  to construct a unique element in $\parsp$
  for a given element in $\datasp$.
  Equation \eqref{eq:mce1} is automatic
  for elements of $\parsp$ since we assume \eqref{eq:assx2}.
  Given an element of $\datasp$,
  equation \eqref{eq:mce3} uniquely determines
  the restriction of $c_{23}^i$ to $x^1 = 0$.
  By the Jacobi identity for vector fields, we necessarily have in $\parsp$,
  \begin{equation} \label{eq:mcx}
    \tsum_{\textnormal{cyclic $(i,j,k)$}} (\tsum_p c_{jk}^p c_{ip}^q + D_i(c_{jk}^q)) = 0
    \end{equation}
  We apply Cauchy Kowalewski to the following data:
  \begin{itemize}
      \item We use $x^1 = 0$ as the initial hypersurface.
    \item The unknown vector is
      $v = (D_2^i, D_3^i, c_{23}^i, a, u, \xi)$.
      Its restriction to $x^1 = 0$ is given by the data
      in $\datasp$ or, in the case of $c_{23}^i$, forced by the data
      as explained above.
    \item The system of differential equations consists of
      equation \eqref{eq:mce2}, 
       the equations for $[D_3,D_1]$ and $[D_1,D_2]$
      in \eqref{eq:mce3},
      and equation \eqref{eq:mcx}.
  \end{itemize}
  So there are 12 scalar unknowns and 12 scalar equations.
  Using $D_2^i(0), D_3^i(0), u(0) > 0$ one can bring these 12 equations
  into the quasilinear form
  \[
    \p_1 v \;=\; A(x,v) + A_2(x,v) \p_2 v
    + A_3 (x,v) \p_3 v
  \]
  with $A$, $A_2$, $A_3$ real analytic in a neighborhood of
  $(0,v(0))$. In fact, \eqref{eq:mce2} 
  can be solved for $\p_1 a$, $\p_1 u$,
  $\p_1 \xi$ which follows from a short calculation.
  By Cauchy Kowalewski there exists a unique solution germ
  for these 12 equations.
  This is not yet the system defining $\parsp$.
  To see that the equation for $[D_2,D_3]$
  in \eqref{eq:mce3} also holds, observe that it holds along
  $x^1 = 0$ by choice of $c_{23}^i$,
  and then everywhere by \eqref{eq:mcx}.
\qed\end{proof}


\section{Maurer-Cartan elements in $\abounce$} \label{sec:bouncesol}

The following theorem has already been referred to.
It is the result of a computer calculation.
It is somewhat remarkable that one has a closed form solution
for elements in $\MC(\abounce)$,
and we only found it after experimentation
with different $\E_G$ spaces, and
different over-parametrizations as in \lemmaref{overp}.

\begin{theorem} \label{theorem:bouncesol}
  Suppose an element of $\mcn$ is given that satisfies
 $u \neq \tfrac{1}{2}, 2$ everywhere,
 with their associated $g_i^{0,1,2,3,4} \in C^\infty(\spatialm,\R)$.
  Then an element of
  \[
    \MC(\abounce) \cap (\theta_0 D_0 + {\abounce^1}_{,G})
  \]
  is given,
  setting $\transition = \tfrac{1}{2}(1+\tanh t)$
  and using the over-parametrization \lemmaref{overp}, by
\begin{align*}
\mu_1  \;&=\;  -\tfrac{1}{2} \log(2\cosh t) &
\gamma_1^2  \;&=\;  g_1^2
                    +\tfrac{2 \mathbf{A}_3}{(u-2) (2 u-1)} \chi\\
\mu_2  \;&=\;  -\tfrac{tu}{2} + \tfrac{1}{2} \log(2\cosh t) &
\gamma_2^2  \;&=\;  g_2^2\\
\mu_3  \;&=\;  -\tfrac{t}{2u} + \tfrac{1}{2} \log(2\cosh t) &
\gamma_3^2  \;&=\;  g_3^2 \displaybreak[0]\\
\beta_1  \;&=\;  D_1  - \tfrac{2 g_3^4}{u-2} \chi D_2 + \tfrac{2 u g_2^3}{2 u-1} \chi D_3 &
\gamma_1^3  \;&=\;  g_1^3
                    -\tfrac{2\mathbf{A}_4}{(u-2) (2 u-1)^2} \chi
                   -\tfrac{8 u g_2^3 g_3^4}{(u-2) (2 u-1)} \chi^2\\
\beta_2    \;&=\; D_2 &
\gamma_2^3  \;&=\;  g_2^3-\tfrac{4 u g_2^3}{2 u-1} \chi\\
\beta_3    \;&=\; D_3 &
\gamma_3^3  \;&=\;  g_3^3+\tfrac{4 g_3^4}{u-2} \chi \displaybreak[0]\\
\gamma_1^0  \;&=\;  \tfrac{1}{2}-\chi &
\gamma_1^4  \;&=\;  g_1^4
                   -\tfrac{2 \mathbf{A}_5}{(u-2)^2 (2 u-1)} \chi
                     -\tfrac{8 u g_2^3 g_3^4}{(u-2) (2 u-1)} \chi^2 \\
\gamma_2^0  \;&=\;  -\tfrac{1}{2} (1+u)+\chi &
\gamma_2^4  \;&=\;  g_2^4-\tfrac{4 u g_2^3}{2 u-1} \chi\\
\gamma_3^0  \;&=\;  -\tfrac{1}{2} (1+\tfrac{1}{u})+\chi &
\gamma_3^4  \;&=\;  g_3^4+\tfrac{4 g_3^4}{u-2} \chi \displaybreak[0]\\
\gamma_1^1  \;&=\;  1 &
\gamma_1^6  \;&=\;  0\\
\gamma_2^1  \;&=\;  g_2^1
                    +\tfrac{\mathbf{A}_1}{(u-2)^2} \chi
                    +\tfrac{4 (g_3^4)^2}{(u-2)^2} \chi^2 &
\gamma_2^6  \;&=\;  0\\
\gamma_3^1  \;&=\;  g_3^1
                    +\tfrac{\mathbf{A}_2}{(2 u-1)^2} \chi
                    +\tfrac{4 u^2 (g_2^3)^2}{(2 u-1)^2} \chi^2 &
\gamma_3^6  \;&=\;  0
\end{align*}
and
\begin{align*}
  \gamma_1^5  \;&=\;  (-\tfrac{4 u g_2^3 g_3^4}{(u-2) (2 u-1)} \chi
-\tfrac{\mathbf{A}_6}{2 u (1+u^2)^2}
+\tfrac{\mathbf{A}_7}{2 u^2 (1+u^2)} t) \sech t
\\
\gamma_2^5  \;&=\;  -\tfrac{2 u g_2^3}{2 u-1} \sech t\\
\gamma_3^5  \;&=\;  \tfrac{2 g_3^4}{u-2} \sech t
\end{align*}
with the abbreviations $\mathbf{A}_1,\ldots,\mathbf{A}_7 \in C^{\infty}(\spatialm,\R)$ given by
\begin{align*}
  \mathbf{A}_1 & = -2 D_3(g_3^4)+u D_3(g_3^4)-D_3(u) g_3^4-2 g_3^3 g_3^4+u g_3^3 g_3^4-2 (g_3^4)^2+u (g_3^4)^2\displaybreak[0]\\
  \mathbf{A}_2 & = -u D_2(g_2^3)+2 u^2 D_2(g_2^3)-D_2(u) g_2^3+u (g_2^3)^2-2 u^2 (g_2^3)^2\\
  & \qquad +u g_2^3 g_2^4-2 u^2 g_2^3 g_2^4\displaybreak[0]\\
  \mathbf{A}_3 & = -2 u g_2^3 g_3^2+u^2 g_2^3 g_3^2+g_2^2 g_3^4-2 u g_2^2 g_3^4\displaybreak[0]\\
  \mathbf{A}_4 & = 2 u D_3(g_2^3)-5 u^2 D_3(g_2^3)+2 u^3 D_3(g_2^3)+2 D_3(u) g_2^3-u D_3(u) g_2^3\\
  & \qquad +2 u g_2^3 g_3^4-5 u^2 g_2^3 g_3^4+2 u^3 g_2^3 g_3^4-g_2^4 g_3^4+4 u g_2^4 g_3^4-4 u^2 g_2^4 g_3^4\displaybreak[0]\\
  \mathbf{A}_5 & = 2 D_2(g_3^4)-5 u D_2(g_3^4)+2 u^2 D_2(g_3^4)+4 u g_2^3 g_3^3-4 u^2 g_2^3 g_3^3+u^3 g_2^3 g_3^3\\
  & \qquad +D_2(u) g_3^4-2 u D_2(u) g_3^4-2 g_2^3 g_3^4+5 u g_2^3 g_3^4-2 u^2 g_2^3 g_3^4\displaybreak[0]\\
  \mathbf{A}_6 & = -4 u^3 D_1(u)+2 u D_2(D_3(u))-2 u^3 D_2(D_3(u))+u^2 D_2(g_3^2)+u^4 D_2(g_3^2)\\
  & \qquad +u^2 D_2(g_3^4)+u^4 D_2(g_3^4)-4 D_2(u) D_3(u)+u^2 D_3(g_2^2)+u^4 D_3(g_2^2)\\
  & \qquad -u^2 D_3(g_2^3)-u^4 D_3(g_2^3)+2 u^2 g_1^3+2 u^4 g_1^3+2 u^2 g_1^4+2 u^4 g_1^4\\
  & \qquad +2 u^3 D_3(u) g_2^2-2 u^3 D_3(u) g_2^3-2 u^3 D_2(u) g_3^2+u^2 g_2^3 g_3^2+u^4 g_2^3 g_3^2\\
  & \qquad -u^2 g_2^4 g_3^2-u^4 g_2^4 g_3^2+2 u D_2(u) g_3^3-2 u^3 D_2(u) g_3^3+u^2 g_2^2 g_3^3+u^4 g_2^2 g_3^3\\
  & \qquad -u^2 g_2^3 g_3^3-u^4 g_2^3 g_3^3-2 u D_2(u) g_3^4-u^2 g_2^2 g_3^4-u^4 g_2^2 g_3^4\\
  & \qquad -2 u^2 g_2^3 g_3^4-2 u^4 g_2^3 g_3^4-u^2 g_2^4 g_3^4-u^4 g_2^4 g_3^4\displaybreak[0]\\
  \mathbf{A}_7 & = 2 u^3 D_1(u)-u D_2(D_3(u))+u^3 D_2(D_3(u))+2 D_2(u) D_3(u)-u^3 D_3(u) g_2^2\\
  & \qquad +u^3 D_3(u) g_2^3+u^3 D_2(u) g_3^2-u D_2(u) g_3^3+u^3 D_2(u) g_3^3+u D_2(u) g_3^4
\end{align*}
\end{theorem}
\begin{proof}
  By direct computer calculation, using the constraint equations
  in \defref{xpre}.
  All denominators are nonzero since $u,\,u-2,\,2u-1 \neq 0$ everywhere.
\qed\end{proof}
\begin{remark} \label{remark:xsw}
  This solution satisfies
  $\lim_{t \to -\infty} \gamma_i^{0,1,2,3,4} = g_i^{0,1,2,3,4}$
  with exponentially fast convergence, and as $|t| \to \infty$ we have,
again exponentially quickly,
\begin{align*}
  \mu_1 & \to -\tfrac{1}{2}|t|&
  \mu_2 & \to -\tfrac{tu}{2} + \tfrac{1}{2}|t|&
  \mu_3 & \to -\tfrac{t}{2u} + \tfrac{1}{2}|t|&
  \gamma_i^5 & \to 0
\end{align*}
In this sense, one recovers
the given element of
$\mcn$ as the past limit,
and another element of $\mcp$ as the future limit,
see \lemmaref{outg} and \theoremref{ibm}.
This is convergence in a rather weak sense,
future work will have to strengthen this
in order to construct the speculative 
($L_\infty$) scattering map mentioned in the introduction.
\end{remark}

\section{Some automorphisms of $\pbounce$ and $\pfree$} \label{sec:generality}

Automorphisms provide conceptual background for various statements
that have been made. 
From a broader perspective,
the BKL problem is sufficiently complicated that
fixing a gauge globally is unlikely to succeed,
one will have to invoke automorphisms.
The normalization operator $\nmlz$ in 
\theoremref{ibm} exemplifies this point of view.

\begin{theorem}[The $1+3$ groupoid and isomorphisms] \label{theorem:g13}
  \begin{itemize}
    \item There is a groupoid, the `1+3 groupoid',
      where an object is as in\footnote{%
        The objects are in one-to-one correspondence
      with manifolds $\spatialm \simeq \R^3$.} \defref{fbm},
      and a
      morphism is a map $\Phi : \R \times \spatialm \to \R \times \spatialm'$,
      $(t,x) \mapsto (a(x)t+b(x),\varphi(x))$
together with a bundle isomorphism
given on sections by
$W' \to W$, 
$\sum_i f_i \theta_i \mapsto \sum_i c (f_i \circ \Phi) \theta_i$,
with $a,b,c \in C^\infty(\spatialm,\R)$ and $a,c > 0$ and $\varphi \in \Diff^\infty(\spatialm,\spatialm')$.
\item A functor from the 1+3 groupoid into the category of real gLa is given
  by the construction of $\E$.
  As a special case, morphisms in the 1+3 groupoid with $\varphi
  = \mathbbm{1}$
  are mapped to the isomorphism induced by\footnote{%
  Morphisms with $a=c=1$
  and $b=0$ and general $\varphi$ are clear.
  }:
  \begin{align*}
    f \;&\mapsto\; f \circ \Phi &
    \theta_0 \;& \mapsto\; c\theta_0 &
  \theta_i \; &\mapsto\; c\theta_i\\
       \sigma_0 \;&\mapsto\; \sigma_0 &
       \sigma_i \;& \mapsto\; \sigma_i &
    \sigma_{ij}\;& \mapsto\; \sigma_{ij}\\
    D_0 \;&\mapsto\; \tfrac{1}{a} D_0 &
  X \; & \mapsto\; \mathrlap{X - \tfrac{1}{a} X(at+b) D_0
  - \tfrac{1}{c} X(c) \sigma_0}
  \end{align*}
  for all $f \in \RR$
  and all $X \in \Der(\RR(\spatialm,\R))$ extended by $X(t)=0$.
  As always, the
  basis-dependent inclusion $\Der(\RR) \inj \CDerEnd(W)$ is implicit.
\item This functor induces
  functors into the category of
 real filtered gLa, by the construction of the 1- or 2- or 3-index
  BKL filtrations.
  The corresponding isomorphisms
  of $\E[[s_2,s_3]]$, $\E[[s_1,s_2,s_3]]$
  restrict to isomorphisms of
  $\pfree$, $\pbounce$
  and these, in turn, induce isomorphisms of the associated gradeds
  $\afree$, $\abounce$. 
  \end{itemize}
\end{theorem}
\begin{proof}
  By direct calculation.
\qed\end{proof}
The $1+3$ groupoid in \theoremref{g13} 
restricts the class of automorphisms,
enough so that the BKL filtration can be defined,
and then some more to fix a foliation by timelike lines,
reflecting the structure (rank 1 frame) of MC-elements in
the associated gradeds
about which we perturb.
Informally, of the original 11 gauge degrees of freedom,
only 5 are left in \theoremref{g13},
the remaining 6 are still present as nilpotent automorphisms.

\newcommand{\sepx}{\;\;&\mapsto\;\;}
\begin{lemma}[Some nilpotent automorphisms of $\afree$]
  \label{lemma:diuehgfieh}
      Every $x \in \bigoplus_{\alpha \neq (0,0,0)} G_{\alpha} \afree^0$
      generates a nilpotent
      automorphism $\exp([x,-]) \in \Aut(\afree)$
      that need not preserve the $\alpha$-grading.
      Every such automorphism is
      induced by an $\R[[s_1,s_2,s_3]]$-linear automorphism of $\pfree$.
 In particular, 
  for all $f \in \RR = C^{\infty}(\R \times \spatialm, \R)$
  we have on $\afree^1$, with formulas
  only given for some elements of $\afree^1$,
  for all cyclic $(i,j,k)$:
  \begin{itemize}
\item[(i)]
  The $\RR$-linear automorphism
  $\exp([f s_js_k \sigma_i,-])$
  acts as:
  \begin{align*}
    \theta_0D_0 \sepx
      \theta_0D_0 - s_js_k D_0(f) \theta_0 \sigma_i
                  + s_js_k f \theta_i D_0\\
      \theta_0\sigma_0 + \theta_i\sigma_i \sepx
          \theta_0\sigma_0 + \theta_i\sigma_i
            + s_js_k f(\theta_0 \sigma_i + \theta_i \sigma_0)\\
      \theta_0\sigma_0 + \theta_j\sigma_j \sepx
        \theta_0\sigma_0 + \theta_j\sigma_j
        + s_js_k f(\theta_i\sigma_0 + \theta_j \sigma_{ij})\\
      \theta_0\sigma_0 + \theta_k\sigma_k \sepx
              \theta_0\sigma_0 + \theta_k\sigma_k
        + s_js_k f(\theta_i \sigma_0 - \theta_k \sigma_{ki})
  \end{align*}
    \item[(ii)] The $\RR$-linear automorphism
  $\exp([f s_js_k \sigma_{jk},-])$
  acts as:
  \begin{align*}
    \theta_0D_0 \sepx
      \theta_0D_0 - s_js_k D_0(f) \theta_0 \sigma_{jk}\\
      \theta_0\sigma_0 + \theta_i\sigma_i \sepx
        \theta_0\sigma_0 + \theta_i\sigma_i\\
      \theta_0\sigma_0 + \theta_j\sigma_j \sepx
        \theta_0\sigma_0 + \theta_j\sigma_j
        - s_js_k f(\theta_j\sigma_k + \theta_k \sigma_j)\\
      \theta_0\sigma_0 + \theta_k\sigma_k \sepx
              \theta_0\sigma_0 + \theta_k\sigma_k
        + s_js_k f(\theta_j\sigma_k + \theta_k \sigma_j)
  \end{align*}
  \end{itemize}
  These automorphisms map
  a given element of
  $\theta_0 D_0 + {\afree^1}_{,G}$ back to this space if and only if
  $f$ satisfies, respectively:
  \begin{itemize}
    \item[(i)] $D_0(f) - (\gamma_j^0+\gamma_k^0)f = 0$.
      Equivalently,
      $D_0(e^{-\mu_j-\mu_k}f)=0$
      if $D_0(\mu_i) = \gamma_i^0$ for all $i$.
    \item[(ii)] $D_0(f) = 0$.
  \end{itemize}
\end{lemma}
\begin{proof}
  A preimage of $x$ under the surjection $\pfree^0 \to \afree^0$
  generates a (formal power series)
  automorphism of $\pfree$ by exponentiation
  that induces the nilpotent automorphism generated by $x$.
 The elements $x \in \afree^0$ in (i), (ii)
 satisfy $\exp([x,-]) = \mathbbm{1} + [x,-]$
  when restricted to $\afree^1$
  since $G_{p_1p_2p_3} \afree^1 = 0$
  if two of $p_1p_2p_3$ are $\geq 2$.
  We have not written out
  how these nilpotent automorphisms act on all elements of $\afree^1$,
  for example what $s_i^2 s_j s_k$ terms they generate,
  but such terms are in ${\afree^1}_{,G}$ anyway.
\qed\end{proof}
\begin{remark} \label{remark:gen56}
  By \lemmaref{diuehgfieh},
  every MC-element as in \lemmaref{mcefree}
can be brought, via a nilpotent automorphism of $\afree$,
into the more specific form in \lemmaref{mcefreenf},
provided the $\smash{\gamma_j^0 - \gamma_k^0}$ do not vanish,
which informally means that the MC-element has
three preferred spatial directions.
Use (i) to set the $\beta_i^0$ to zero
which yields a synchronous frame 
and makes the $\gamma_i^6$ zero by \lemmaref{g60};
use (ii) to set the $\gamma_i^5$ to zero,
informally
to align the MC-element with the orthogonal decomposition of $W$
that defines the filtration;
finally choose $\mu_i$ and $D_i$ appropriately.
\end{remark}

\begin{lemma}[Automorphism giving the transformation in \lemmaref{mcpNEW}]
  The transformation (a) in \lemmaref{mcpNEW}
  is given by composing the following automorphisms of $\afree$,
  induced by automorphisms of $\pfree$,
  in this specific order:
  \begin{itemize}
    \item[1)] An automorphism with
      $a = c = A$ and $b = \log B$
      as in \theoremref{g13}.
    \item[2)] The nilpotent 
      $\exp([f_1s_2s_3\sigma_1,-])
      \circ \exp([f_2s_3s_1\sigma_2,-])
      \circ \exp([f_3s_1s_2\sigma_3,-])$ with
      \[
          f_i\;=\;
          e^{(At + \log B)(g_j^0+g_k^0)} D_i(A t + \log B)
      \]
    \item[3)] Reflection by $\sigma \mathbbm{1}$
     (with $\sigma$ as in \lemmaref{mcpNEW})
      of the basis elements $\theta_1,\theta_2,\theta_3$.
  \end{itemize}
\end{lemma}
\begin{proof}
  We use the over-parametrization in \lemmaref{overp} even if
  intermediate results lie outside its range.
  Step 1) uses $a = c$ to keep $\theta_0 D_0$ in place. The choice
  $a = c = A$ brings the $\gamma_i^0$ into the desired form,
  and $b = \log B$ brings the $\beta_i$ into the desired form.
  In Step 2) the $f_i$
  are uniquely determined by requiring that they remove any
  $s_js_k \theta_i D_0$ terms
  that we may have inadvertently introduced in Step 1),
  so now $\beta_i^0 = 0$.
  Any $s_js_k \theta_i \sigma_0$
  that we may have inadvertently introduced in Step 1)
  only appear in the combination $s_js_k(\theta_0 \sigma_i + \theta_i \sigma_0)$,
  so after Step 2)
  we are back in the range of \lemmaref{overp}.
  The $\gamma_i^5$ are zero because
  no $s_js_k(\theta_k \sigma_j + \theta_j \sigma_k)$ was introduced.
  The $\gamma_i^6$ are zero by \lemmaref{g60}.
  Note that Step 2) did not change $\gamma_i^0$ or $\beta_i$.
  Step 3) is clear.
\qed\end{proof}



\appendix
\section{Review of the spectral sequence of a filtered complex}\label{appendix:ssfc}

We review spectral sequences in a concrete form,
suited for our application.
Suppose we have a complex in the most mundane sense, namely
a vector space $V$ and a $d \in \End(V)$ with $d^2 = 0$.
Its homology is $H(d) = \ker d/\image d$.

Suppose a bounded non-increasing filtration $V_{\geq i}$ is fixed
that respects the differential in the sense that $d V_{\geq i} \subset V_{\geq i}$.
For concreteness, suppose the filtration comes from a grading, so
$V_{\geq i} = V_i \oplus V_{i+1} \oplus \ldots \oplus V_P$
for some grading
$V = V_0 \oplus V_1 \oplus \ldots \oplus V_P$.
Then $d$ is a lower triangular matrix
with entries
$d_{ij} \in \Hom(V_j,V_i)$.
For every $p > 0$
let $d_{ip}^+$ be the $p \times p$
submatrix of $d$
with upper left corner $d_{ii}$.
For example
\[
  d_{22}^+ = \begin{pmatrix}
    d_{22} & 0 \\ d_{32} & d_{33}
  \end{pmatrix}
  \qquad
  d_{23}^+ = \textstyle\begin{pmatrix}
    d_{22} & 0 & 0 \\ d_{32} & d_{33} & 0 \\
    d_{42} & d_{43} & d_{44}
  \end{pmatrix}
\]
Likewise, let $d_{ip}^-$ be the $p\times p$ submatrix with lower right corner $d_{ii}$.
\begin{lemma}[The spectral sequence]
Set $Z_{i0}=V_i$ and $B_{i0}=0$
and for $p > 0$ set
\begin{align*}
    Z_{ip}
    \;&=\;
    \{ x \in V_i \,\mid\,
      \exists \ast
      \in V_{i+1} \oplus \ldots \oplus V_{i+p-1}:
      \;\; (\begin{smallmatrix} x \\ \ast\end{smallmatrix}) \in \ker d_{ip}^+ \}\\
    B_{ip}
    \;&=\;
    \{ x \in V_i \,\mid\,
      (\begin{smallmatrix} 0 \\ x \end{smallmatrix}) \in \image d_{ip}^- \}
\end{align*}
Then $B_{ip} \subset Z_{ip} \subset V_i$. Define $\bullet_{ip} = Z_{ip}/B_{ip}$.
A well-defined $D_{i+p,i} \in \Hom(\bullet_{ip}, \bullet_{i+p,p})$
is given by $D_{ii} = d_{ii}$ if $p=0$,
and if $p > 0$ induced by
the map $Z_{ip} \to \bullet_{i+p,p}$ given by
\[
x
\;\mapsto\;
(d_{i+p,i}\, \ldots \, d_{i+p,i+p-1})
(\begin{smallmatrix} x \\ \ast \end{smallmatrix})
\]
with any filler $\ast$ as in the definition of $Z_{ip}$.
For every fixed $p$,
the vector space $\bigoplus_i \bullet_{ip}$
together with the maps
$(D_{i+p,i})_i$
(combined in a block matrix with entries only along the $p$-th subdiagonal)
is a complex. 
The homology at $\bullet_{ip}$ is equal to $\bullet_{i,p+1}$.
\end{lemma}
\begin{proof}
  By working out the ramifications of $d^2 = 0$.
\qed\end{proof}
Even though one
iteratively takes homology,
each $\bullet_{ip}$ is a concrete space, namely a subquotient of $V_i$.
This is because,
a subquotient of a subquotient is a subquotient.
\begin{theorem}[Main fact] \label{theorem:specseq}
  Let $H(d)_{\geq i} \subset H(d)$ be all elements
  that have a representative in $V_{\geq i}$.
  Then the associated graded, $\Gr H(d)$,
  is isomorphic, as a graded vector space, to the last page
  $\bigoplus_i \bullet_{i,P+1}$
  of the spectral sequence.
\end{theorem}

\begin{proof}
The map $Z_{i,P+1} \to H(d)_{\geq i} / H(d)_{\geq i+1}$ defined by
$x \mapsto (\begin{smallmatrix} x \\ \ast \end{smallmatrix})$,
with $\ast$ any filler as in the definition of $Z_{i,P+1}$,
is independent of $\ast$, surjective, and has kernel $B_{i,P+1}$.
\qed\end{proof}
 The definition of $\bullet_{ip}$ given above makes the proof of
 \theoremref{specseq} straightforward.
 In practice, it is important to know the differentials on the first few pages.
\begin{lemma}[The first few pages]
  Choose $h_i \in \End(V_i)$ such that
  $d_{ii} h_i d_{ii} = d_{ii}$.
  Then:
  \begin{itemize}
    \item $D_{ii} = d_{ii}$.
    \item $D_{i+1,i}$ is induced by $d_{i+1,i}$.
    \item $D_{i+2,i}$ is induced by $d_{i+2,i} - d_{i+2,i+1} h_{i+1} d_{i+1,i}$.
  \end{itemize}
\end{lemma}
\begin{proof}
  Omitted.
  \qed\end{proof}
  If $V$ carries an additional homological grading,
  the construction is compatible with it.
This is useful in practice,
but irrelevant for the logic of spectral sequences.

\end{document}